\documentclass[letterpaper, 10 pt, conference]{ieeeconf}  % Comment this line out if you need a4paper

\IEEEoverridecommandlockouts                              % This command is only needed if 
\usepackage{graphics} % for pdf, bitmapped graphics files
\usepackage{epsfig} % for postscript graphics files
\usepackage{mathptmx} % assumes new font selection scheme installed
\usepackage{times} % assumes new font selection scheme installed
\usepackage{amsmath} % assumes amsmath package installed
\usepackage{amssymb}  % assumes amsmath package installed
\usepackage{algorithm}
\usepackage{algpseudocode}
\usepackage{mathtools}
\usepackage{bbm}
\usepackage{tikz}
\usepackage{layouts}
\usepackage{pgfplots} 
\usepackage{tikzscale}
\usepackage{subcaption}
\usepackage{graphicx}
\usepackage{centernot}

\newtheorem{prop}{Proposition}
\newtheorem{theo}{Theorem}
\newtheorem{ass}{Assumption} 
\newtheorem{rem}{Remark}
\newtheorem{defi}{Definition}

\title{\LARGE \bf
Stability Mechanisms for Predictive Safety Filters
}

\author{Elias Milios$^{1}$, Kim Peter Wabersich$^{1}$, Felix Berkel$^{1}$, and Lukas Schwenkel$^{2}$% <-this % stops a space
\thanks{$^{1}$Elias Milios, Kim Peter Wabersich, and Felix Berkel are with the Corporate Research of Robert Bosch GmbH, 71272 Renningen, Germany.
        \{Elias.Milios, KimPeter.Wabersich, Felix.Berkel\}@de.bosch.com}%
\thanks{$^{2}$Lukas Schwenkel is with the Institute for Systems Theory and Automatic Control, University of Stuttgart, 70550 Stuttgart, Germany.
        Lukas.Schwenkel@ist.uni-stuttgart.de}%
}

\begin{document}

\maketitle
\thispagestyle{empty}
\pagestyle{empty}

%%%%%%%%%%%%%%%%%%%%%%%%%%%%%%%%%%%%%%%%%%%%%%%%%%%%%%%%%%%%%%%%%%%%%%%%%%%%%%%%
\begin{abstract}
    Predictive safety filters enable the integration of potentially unsafe learning-based control approaches and humans into safety-critical systems. In addition to simple constraint satisfaction, many control problems involve additional stability requirements that may vary depending on the specific use case or environmental context. In this work, we address this problem by augmenting predictive safety filters with stability guarantees, ranging from bounded convergence to uniform asymptotic stability. The proposed framework extends well-known stability results from model predictive control (MPC) theory while supporting commonly used design techniques. As a result, straightforward extensions to dynamic trajectory tracking problems can be easily adapted, as outlined in this article. The practicality of the framework is demonstrated using an automotive advanced driver assistance scenario, involving a reference trajectory stabilization problem.
\end{abstract}

%%%%%%%%%%%%%%%%%%%%%%%%%%%%%%%%%%%%%%%%%%%%%%%%%%%%%%%%%%%%%%%%%%%%%%%%%%%%%%%%
\section{INTRODUCTION}
Advances in the field of learning-based control and the increasing demand for human-machine interaction are driving the need for modular safety certificates in control systems.
Prominent domains include automated driving \cite{vehicle_automation}, smart factories \cite{human_robot_collab}, or surgical robotics \cite{surg_robotics}. Predictive safety filter methods address this challenge by ensuring constraint satisfaction using MPC techniques. While safety in the form of constraint satisfaction is often a basic specification, many modern control problems still require classical stability properties, e.g., to reduce stress on actuators. These stability properties can vary depending on the underlying system, use case, and environment. For example, during the development of an advanced driver assistance system, different stability characteristics are desired for different phases. In a data collection process for system identification or controller tuning, safety in terms of constraint satisfaction may be sufficient, see Fig.~\ref{fig:stability properties}, left. However, in an automated driver assistance task, such as lane keeping or automotive cruise control, additional requirements like convergence or asymptotic stability with respect to a reference provided, e.g., computed by a higher-level planner, are desirable, see Fig.~\ref{fig:stability properties}, center and right.
   \begin{figure*}[t] 
      \centering
        \includegraphics[scale=0.65]{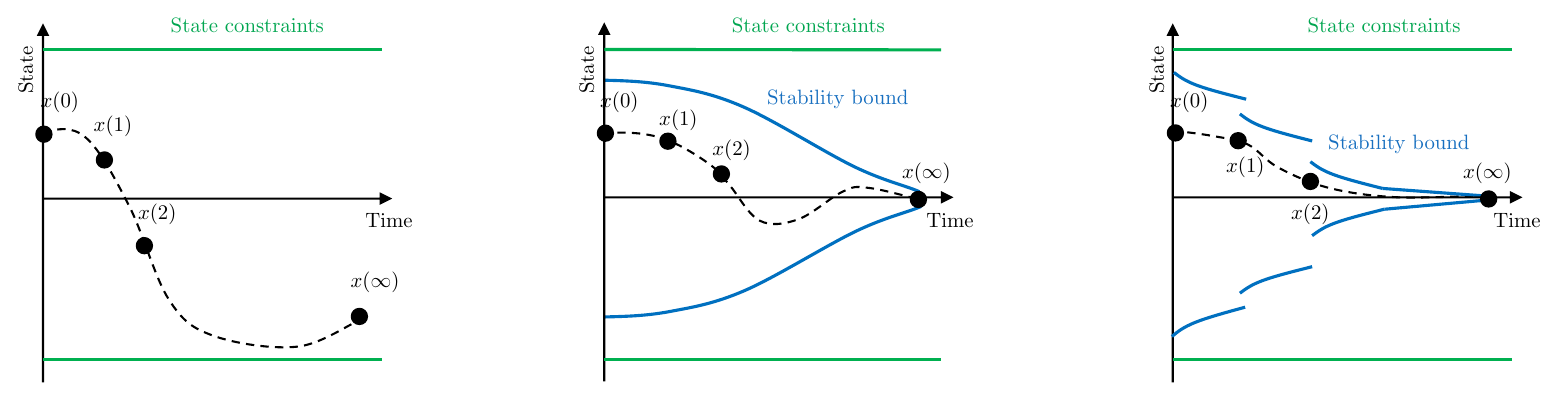}
        \caption{Illustration of the stability mechanisms considered in this work. \textbf{Left:} Safety in terms of constraint satisfaction. \textbf{Center:} Safe stability in terms of a bounded and converging closed-loop state evolution, which satisfies the constraint. \textbf{Right:} Safe stability in terms of an uniformly asymptotically stable closed-loop state evolution, which satisfies the constraints.}
        \label{fig:stability properties}
   \end{figure*}
In this work, we provide a modular safety filter layer that ensures constraint satisfaction and stability specifications. This enables the integration of learning-based controllers and humans into safety-critical systems.

\textit{Related work:}  
The concept of control barrier functions (CBF) provides combined constraint satisfaction and asymptotic stability in the sense of Lyapunov~\cite{CBF_Agrawal}. The main challenge in using these techniques is the explicit design of a control barrier and a control Lyapunov function (CLF), which is valid on a possibly large domain satisfying all safety constraints. This design task is especially difficult for nonlinear systems with higher dimensional state and input spaces. In addition, stability with respect to reference trajectories further complicates the underlying design process, which can be handled naturally in an MPC framework.

In \cite{Soloperto}, a generic framework for changing the objective of an existing MPC scheme to a user-defined learning objective is presented. By introducing a suitable constraint, a performance bound with respect to the original objective is provided. The work in \cite{didier2024predictive} presents an extension of predictive safety filters that ensures Lyapunov stability guarantees with respect to an equilibrium point. This is achieved by an extended state that evolves with respect to a set-valued mapping, as done in suboptimal MPC~\cite{allan,feas_impl_stab}. The design techniques in \cite{didier2024predictive} differ from most MPC literature, complicating straightforward extensions to, e.g., tracking problems.
\par
Our method follows a similar concept as \cite{Soloperto} and \cite{didier2024predictive} with additional notions of stability and corresponding theoretical results, tailored to practical challenges.

\textit{Contribution:} 
We propose a stability-enhanced predictive safety filter that can be combined with arbitrary controllers and is applicable to nonlinear systems (Section~\ref{sec:MPSF}). 
In particular, we provide a \textit{rigorous analysis} that comes with \textit{formal stability guarantees} concerning \textit{different stability specifications} (Section~\ref{sec:theoretical_results}).
The formulation is easily adaptable to a variety of safety-critical control problems, resulting in a \textit{modular framework}. We outline the application to robust and dynamic tracking settings (Section~\ref{sec:Extensions}) and provide a \textit{closed-loop performance bound}.
The design and implementation of the scheme are illustrated using an advanced driver assistance system example, requiring safety and stability with respect to a dynamic reference (Section~\ref{sec:num_example}).

\textit{Notation:}
The quadratic norm with respect to a positive definite matrix $Q \succ 0$ is denoted by $\lVert x \rVert_{Q}^2 = x^\top Q x$. The symbol $\mathbb{I}_{\geq 0}$ denotes the set of nonnegative integers, and the symbol $\mathbb{I}_{[0,N]}$ denotes the set $\{0,1,\ldots,N\}$. For two vectors $a \in \mathbb{R}^n$ and $b \in \mathbb{R}^m$, the ordered pair $(a,b) \in \mathbb{R}^{n+m}$ represents their concatenation. The interior of a set $\mathbb{X}$ is denoted by $\mathrm{int}(\mathbb{X})$. A function $\alpha: \mathbb{R}_{\geq 0} \rightarrow \mathbb{R}_{\geq 0}$ is called a $\mathcal{K}_\infty$ function if it is continuous, strictly increasing, and it holds that $\alpha(0) = 0$ and $\alpha(s)\rightarrow \infty$ as $s \rightarrow \infty$. 

\section{Preliminaries and Problem Statement}
We consider nonlinear discrete-time systems of the form
\begin{equation}\label{eq:nonlinear_system}
    x(k+1) = f(x(k),u(k)),
\end{equation}
with state $x(k) \in \mathbb{R}^n$, control input $u(k) \in \mathbb{R}^m$, and time step $k \in \mathbb{I}_{\geq0}$. We assume that the dynamics $f:\mathbb{R}^n \times \mathbb{R}^m \rightarrow \mathbb{R}^n$ are continuous and that the origin of system \eqref{eq:nonlinear_system} is an equilibrium point, i.e., $0=f(0,0)$. The system is subject to state and input constraints $(x(k),u(k)) \in \mathbb{X} \times \mathbb{U}$, with $\mathbb{X} \subseteq \mathbb{R}^n$ and $\mathbb{U} \subset \mathbb{R}^m$, $(0,0)\in \mathbb{X} \times \mathbb{U}$, and $\mathbb{U}$ compact. 
The overall goal is to apply a potentially unsafe and non-stabilizing desired input $u_\mathrm{des}(k)\in \mathbb{R}^m$ to~\eqref{eq:nonlinear_system}, generated, e.g., by a learning algorithm or a human interacting with the system. To ensure state and input constraint satisfaction at different levels of stability, we aim at deriving a filtering mechanism that monitors $u_\mathrm{des}(k)$ and modifies it if necessary.
To this end, we recap an existing model predictive safety filter formulation for ensuring constraint satisfaction~\cite{Linear_SafetyFilter} and extend it to provide stability guarantees in a second step.
\par
Given a state $x(k)$ and a desired input $u_\mathrm{des}(k)$, we consider a predictive safety filter that is based on solving a receding horizon predictive control problem of the form:
\begin{subequations}\label{eq:saftey_filter}
\begin{align}
    \min_{u_{\cdot|k}} \ & G(x_{\cdot|k},u_{\cdot|k},u_{\mathrm{des}}(k))\label{eq:safe_filter_cost}\\
    \text{s.t. } \  & x_{0|k} = x(k) \label{eq:safety_filter_init_constr}\\
    & x_{i+1|k} = f(x_{i|k},u_{i|k}), \  \forall i \in \mathbb{I}_{[0,N-1]} \label{eq:safety_filter_dynamics}\\
    & x_{i|k} \in \mathbb{X}, \  \forall i \in \mathbb{I}_{[0,N-1]} \\
    & u_{i|k} \in \mathbb{U}, \ \forall i \in \mathbb{I}_{[0,N-1]} \label{eq:safety_filter_input_constr}\\
    & x_{N|k} \in \mathbb{X}_\mathrm{f}\label{eq:safety_filter_term_constr}.
\end{align}
\end{subequations}
Here, $x_{i|k}$ and $u_{i|k}$ denote the $i$-th predicted state and input at time $k$ and $x_{\cdot|k}=\{x_{i|k}\}_{i=0}^N$ denotes a predicted open-loop state evolution of system \eqref{eq:nonlinear_system} at time $k$, given $x_{0|k}=x(k)$ and a predicted input sequence $u_{\cdot|k} = \{u_{i|k}\}_{i=0}^{N-1}$. The last predicted state $x_{N|k}$ is constrained to the terminal set $\mathbb{X}_\mathrm{f}$, with $N \in \mathbb{I}_{>0}$ denoting the prediction horizon. 
Given $x(k)$, we define the set of \textit{feasible input sequences} $u_{\cdot|k}$ for problem \eqref{eq:saftey_filter} by
\begin{equation}
    \mathcal{U}(x(k)):=\{u_{\cdot|k} \mid \eqref{eq:safety_filter_init_constr} - \eqref{eq:safety_filter_term_constr} \},
\end{equation}
and the set of \textit{feasible} states $x(k)$ for problem \eqref{eq:saftey_filter} by
\begin{equation}
    \mathcal{X}:= \{x(k) \mid \exists u_{\cdot|k} \in \mathcal{U}(x(k)) \}.
\end{equation}

Safety of a desired input $u_\mathrm{des}(k)$ is ensured by solving~\eqref{eq:saftey_filter} and applying a possibly modified input which allows to ensure existence of a \textit{safe backup trajectory} towards $\mathbb{X}_{\mathrm{f}}$. By designing $\mathbb{X}_\mathrm{f}$ to be forward invariant under some safe controller (compare with \cite{Linear_SafetyFilter}), the computed safe backup implies recursive feasibility of \eqref{eq:saftey_filter} and with this constraint satisfaction of the operated systems for all times. Selecting $G~:~\mathbb{R}^{Nn} \times \mathbb{R}^{Nm} \times \mathbb{R}^m \rightarrow \mathbb{R}$, e.g., as
$G(x_{\cdot|k},u_{\cdot|k},u_{\mathrm{des}}):= \lVert u_\mathrm{des}(k) - u_{0|k}\rVert^2$ yields the desired filtering property.

While the structure of problem \eqref{eq:saftey_filter} ensures safety in terms of constraint satisfaction, it does not provide any additional stability properties. In the following, we define two different `levels' of stability, characterizing the desired stability enhancements to \eqref{eq:saftey_filter}, which are presented in the following sections.

  \begin{defi}\label{defi:stability}
     Let $x^* = 0$ be an equilibrium point for system \eqref{eq:nonlinear_system} with time-varying feedback law $u(k) = \mu(x(k),k)$, where $\mu:  \mathbb{R}^n \times \mathbb{I}_{\geq 0} \rightarrow \mathbb{R}^m$. We call the resulting closed-loop system:
     \par 
     \textit{(1) Bounded and converging} to $x^*$ if, for each $\varepsilon > 0$, there exists a $\delta = \delta(\varepsilon)>0$ such that $\lVert x(0) \rVert < \delta \Rightarrow \lVert x(k) \rVert < \varepsilon$ for all $k \in \mathbb{I}_{\geq 0}$, and $ x(k) \rightarrow 0$ as $k \rightarrow \infty$.
    \par 
     \textit{(2) Uniformly asymptotically stable} with respect to $x^*$ if, for each $\varepsilon > 0$, there exists a $\delta = \delta(\varepsilon)>0$ such that $\lVert x(k_0) \rVert < \delta \Rightarrow \lVert x(k) \rVert < \varepsilon$ for all $k_0,k \in \mathbb{I}_{\geq 0}$ with $k \geq k_0$, and $ x(k) \rightarrow 0$ as $k \rightarrow \infty$, compare with \cite[Definition 4.2]{Elaydi2005}.
 \end{defi}

An illustration of the implications of Definition~\ref{defi:stability} on the closed-loop behavior is provided in Fig. \ref{fig:stability properties}.
For bounded convergence (Fig.~\ref{fig:stability properties}, center), the closed-loop state evolution must be bounded with respect to the initial state $x(k_0)$ with initial time $k_0=0$ and converge to $x^*$ as time goes to infinity. This implies stronger guarantees compared to plain constraint satisfaction while maintaining a certain degree of freedom for the system evolution. This can be desired, e.g., for system identification or learning-based settings with the goal of gathering real-world data with multiple rollouts. In such settings, properties in the sense of bounded convergence can be desired to ensure that the system stays inside a certain area of the state space and eventually converges to a desired equilibrium point in order to start a new rollout.
In contrast, for uniform (and also non-uniform) asymptotic stability (Fig.~\ref{fig:stability properties}, right), the closed-loop state evolution must be bounded with respect to any time $k_0 \in \mathbb{I}_{\geq 0}$ and any initial state $x(k_0)$ and converge to $x^*$ as time goes to infinity. This definition implies stronger requirements on the closed-loop state evolution, e.g. requiring the system to stay in $x^*$ once it is reached. 
We introduce the main concepts used to extend \eqref{eq:saftey_filter} with stability properties according to Definition~\ref{defi:stability} in Section \ref{sec:MPSF}. The underlying theory for stability in terms of bounded convergence and uniform asymptotic stability is derived in Section~\ref{sec:convergence} and \ref{sec:stability}, respectively. In Section~\ref{sec:Extensions}, we extend the framework from stabilization of the origin to stabilization of dynamic references, and apply the framework to an advanced driver assistance scenario in Section \ref{sec:num_example}.

\section{Stabilizing Predictive Safety Filters}\label{sec:MPSF}
To enhance \eqref{eq:saftey_filter} such that stability properties can be guaranteed on top of safety, the main idea is to constrain a \textit{stability cost} of the form:
\begin{equation}\label{eq:stab_cost}
    J(x(k),u_{\cdot|k}):=\sum_{i=0}^{N-1} \ell (x_{i|k},u_{i|k}) + V_\mathrm{f}(x_{N|k}),
\end{equation} 
 with a \textit{stability bound} $J_\mathrm{B}: \mathbb{I}_{\geq 0} \rightarrow \mathbb{R}_{\geq 0}$ to additionally enforce $J(x(k),u_{\cdot|k})\leq J_{\mathrm{B}}(k)$ in \eqref{eq:saftey_filter}. This ensures that safe solutions $u_{\cdot|k} \in \mathcal{U}(x(k))$ provide a stability cost that is not larger than the current stability bound. The stability cost is defined similarly to a standard stabilizing MPC cost function by combining a stage cost $\ell:\mathbb{R}^n\times\mathbb{R}^m \rightarrow \mathbb{R}_{\geq 0}$ with an appropriately chosen terminal cost $V_\mathrm{f}:\mathbb{R}^n \rightarrow \mathbb{R}_{\geq 0}$ compensating for the neglected infinite horizon tail of the stage cost \cite{MPC_rawlings}.     
\begin{ass}\label{ass:cont_stagecost}
 The stage and terminal cost $\ell$ and $V_\mathrm{f}$ are continuous satisfying $\ell(0,0) = V_\mathrm{f}(0) = 0$ and there exists an $\alpha_\ell\in \mathcal{K}_{\infty}$ such that $\ell(x,u) \geq \alpha_\ell(\lVert x \rVert)$ for all $(x,u) \in \mathbb{X} \times \mathbb{U}$.
\end{ass}

 Based on this, the main idea is to shape the desired closed-loop behavior through suitable design choices of $J(\cdot)$ and $J_\mathrm{B}(k)$. Given a state $x(k)$ and desired input $u_\mathrm{des}(k)$ at time $k$ define the \textit{stability-enhanced safety filter problem} $\mathbb{P}(k)$ by
\begin{subequations}\label{eq:stability_filter}
\begin{align}
    \min_{u_{\cdot|k}} \ & \eqref{eq:safe_filter_cost}\\
    \text{s.t. } \  & \eqref{eq:safety_filter_init_constr} - \eqref{eq:safety_filter_term_constr} \label{eq:stab_filter_safety_filter_parts}\\
    & J(x(k),u_{\cdot|k}) \leq J_\mathrm{B}(k) \label{eq:stability_constraint},
\end{align}
\end{subequations}
where we include the \textit{stability constraint} \eqref{eq:stability_constraint} compared to~\eqref{eq:saftey_filter}.
 We denote the set of \textit{feasible input sequences} $u_{\cdot|k}$, given $x(k)$ and $J_\mathrm{B}(k)$, by
\begin{equation}
    \mathcal{U}_\mathbb{P}(x(k),J_\mathrm{B}(k)):= \{ u_{\cdot|k} \mid \eqref{eq:stab_filter_safety_filter_parts}-\eqref{eq:stability_constraint} \},
\end{equation}
and the set of \textit{feasible states} $x(k)$ by
\begin{equation}
    \mathcal{X}_{\mathbb{P}}(J_\mathrm{B}(k)):= \{x(k)  \mid  \exists u_{\cdot|k} \in \mathcal{U}_\mathbb{P}(x(k),J_\mathrm{B}(k)) \}.
\end{equation}
With this, for all $k \in \mathbb{I}_{\geq 0}$, $x(k)$, and $J_\mathrm{B}(k)$, we have that $\mathcal{X}_{\mathbb{P}}(J_\mathrm{B}(k)) \subseteq \mathcal{X}$ and $\mathcal{U}_\mathbb{P}(x(k),J_\mathrm{B}(k)) \subseteq \mathcal{U}(x(k))$. The solution to \eqref{eq:stability_filter}, given $x(k)$, $u_\mathrm{des}(k)$, and $J_\mathrm{B}(k)$, is the optimal input sequence $u^*_{\cdot|k}$ and the so-called \textit{closed-loop stability cost} defined by
\begin{equation}
    V(x(k),k)=J(x(k),u^*_{\cdot|k}).
\end{equation}
The resulting input applied to system \eqref{eq:nonlinear_system} is defined by the first element of the optimal solution, i.e., $u(k):=u^*_{0|k}$.

Based on this, the overall resulting scheme is presented in Algorithm~\ref{algo:nom_stabfilter_nom}.
\begin{algorithm}[t]
    \caption{Model Predictive Stability Filter}\label{algo:nom_stabfilter_nom}
    \begin{algorithmic}[1]
    \Statex \textbf{Input}: $\mathbb{P}$, $N$, $J_\mathrm{B}(0)$. \textbf{Output}: $u(k) = \mu(x(k),k)$.
    \State Set $k:=0$.
    \State Measure current state $x(k)$, obtain desired input $u_\mathrm{des}(k)$.\label{algo:step_repeat2}
    \State Solve $\mathbb{P}$ with $x(k)$, $u_\mathrm{des}(k)$, and $J_\mathrm{B}(k)$ and obtain $u^*_{\cdot|k}$.
    \State Apply $u(k)=\mu(x(k),k)=u^*_{0|k}$.
    \State Construct $J_\mathrm{B}(k+1)$ compliant with Assumption \ref{ass:evol_stab_bound}.
    \State Set $k:=k+1$ and go back to \ref{algo:step_repeat2}.
    \end{algorithmic}
\end{algorithm}
Consider a current state $x(k)$, a given stability bound $J_\mathrm{B}(k)$, a desired input $u_\mathrm{des}(k)$. To analyze the influence of $u_\mathrm{des}(k)$ on the safety and stability of system~\eqref{eq:nonlinear_system}, we test whether $u_\mathrm{des}(k)$ evolves the system to a state $x(k+1)$, for which we can construct a stabilizing input sequence $u_{\cdot|k} \in \mathcal{U}_\mathbb{P}(x(k),J_\mathrm{B}(k))$. If the test is successful, $u_\mathrm{des}(k)$ is certified as \textit{safe and stabilizing} and is applied to the system. Otherwise, $u_{\mathrm{des}}(k)$ is minimally modified with respect to \eqref{eq:safe_filter_cost} in order to satisfy all constraints. At the subsequent time step, we repeat the computation given a new desired input $u_\mathrm{des}(k+1)$ and a new stability bound $J_\mathrm{B}(k+1)$. To ensure feasibility of problem \eqref{eq:stability_filter} at time $k+1$ as well as a closed-loop state evolution compliant with the desired stability specification, we require the successor stability bound $J_\mathrm{B}(k+1)$ to satisfy the following Assumption.
\begin{ass}\label{ass:evol_stab_bound} Consider $\zeta_{\mathrm{min}}>0$ and let $\rho \in [\zeta_{\mathrm{min}},1]$. We assume that, for all $k \in \mathbb{I}_{\geq 0}$, $x(k) \in \mathcal{X}_\mathbb{P}(J_\mathrm{B}(k))$, and $u^*_{\cdot|k} \in \mathcal{U}_\mathbb{P}(x(k),J_\mathrm{B}(k))$, there exists a candidate input sequence $u^\mathrm{c}_{\cdot|k+1} \in \mathcal{U}(x(k+1))$ and $\zeta(k+1) \in [\zeta_\mathrm{min},\rho]$ such that
\begin{multline}\label{eq:evolu_stab_bound}
J(x(k+1),u^\mathrm{c}_{\cdot|k+1}) \leq V(x(k),k)  -\rho \ell(x(k),u^*_{0|k}) \\ \leq J_\mathrm{B}(k+1) \leq V(x(k),k) -\zeta(k+1) \ell(x(k),u^*_{0|k}).
\end{multline}
\end{ass}

With Assumption \ref{ass:evol_stab_bound}, we constrain the successor stability bound $J_\mathrm{B}(k+1)$ from both sides. On the one hand, Assumption~\ref{ass:evol_stab_bound} requires $J_\mathrm{B}(k+1)$ to enforce a minimal decrease of the closed-loop stability cost, i.e., to enforce $V(x(k+1),k+1)-V(x(k),k) \leq 0$, necessary to establish convergence of the closed-loop system. The rate at which $V(x(k),k)$ is enforced to decrease is determined by $\zeta(k+1)$, which is limited from below by the constant $\zeta_{\mathrm{min}}$ and from above by the constant $\rho$. On the other hand, we assume a lower bound on $J_\mathrm{B}(k+1)$ to prevent a rapid decrease of $J_\mathrm{B}(k+1)$, which renders problem \eqref{eq:stability_filter} infeasible. This lower bound is given by the stability cost resulting from a feasible candidate solution for time $k+1$, which provides a decrease of the closed-loop stability cost with decay factor $\rho$.

While Assumption \ref{ass:evol_stab_bound} implies feasibility of problem~\eqref{eq:stability_filter} and a monotonic decrease of $V(x(k),k)$, further requirements on the ingredients of \eqref{eq:stability_filter}, namely $\ell$, $V_\mathrm{f}$, $\mathbb{X}_\mathrm{f}$, and $J_\mathrm{B}$, are needed to align the closed-loop behavior with the desired stability specifications from Definition \ref{defi:stability}. 

\section{Theoretical Results and Design}\label{sec:theoretical_results}
This section introduces additional technical assumptions and provides design choices for the stability filter problem~\eqref{eq:stability_filter} so that constraint satisfaction, recursive feasibility, and the desired stability specification can be guaranteed. We first establish bounded convergence according to Definition \ref{defi:stability},~(1), followed by possible design choices to satisfy the underlying assumptions based on established MPC techniques. These assumptions and results are strengthened to meet stability guarantees according to Definition \ref{defi:stability},~(2), where possible design choices are derived from suboptimal MPC literature.

\subsection{Bounded convergence}\label{sec:convergence}
In order to establish the stability specification according to Definition \ref{defi:stability},~(1), we need additional assumptions on the initial stability bound $J_\mathrm{B}(0)$ and on the stability cost terms. We define a lower bound using the optimal value function of a \textit{corresponding} MPC problem, which is given by
\begin{equation}
    V^\mathrm{MPC}(x(k))~:~=~\min_{u_{\cdot|k} \in \mathcal{U}(x(k))} J(x(k),u_{\cdot|k}).
\end{equation}
\begin{ass}\label{ass:init_stab_bound}
 There exists an $\alpha_2\in \mathcal{K}_\infty$ such that for any $x(0) \in \mathcal{X}$, the initial stability bound $J_\mathrm{B}(0)$ satisfies
 \begin{equation}
     V^\mathrm{MPC}(x(0)) \leq J_\mathrm{B}(0) \leq \alpha_2(\lVert x(0) \rVert).
 \end{equation}
\end{ass}

The lower bound ensures that a feasible solution to \eqref{eq:stability_filter} exists for any given $x(0) \in \mathcal{X}$. The upper bound is required to establish stability according to Definition~\ref{defi:stability},~(1) as follows.
\begin{theo}\label{theo:general_stab_filter}
  Let Assumptions \ref{ass:cont_stagecost}~-~\ref{ass:init_stab_bound} hold and suppose that $\mathcal{X}$ contains a neighborhood of the origin. Then application of Algorithm \ref{algo:nom_stabfilter_nom} with $\mathbb{P} := \mathbb{P}(k)$ yields bounded convergence according to Definition \ref{defi:stability},~(1) with respect to $x^*=0$ for any $x(0) \in \mathcal{X}$ and all $u_\mathrm{des}(k) \in \mathbb{R}^m$ with $k\in \mathbb{I}_{\geq 0}$.
\end{theo}
\begin{proof}
    In the following, we first prove that \eqref{eq:stability_filter} is recursively feasible following initial feasibility at time $k=0$. Following this, we conclude bounded convergence.
    
    \textit{Recursive feasibility:}
    We prove recursive feasibility by induction. For the induction start, by Assumption \ref{ass:init_stab_bound}, for any $x(0) \in \mathcal{X}$, we can find $u_{\cdot|0} \in \mathcal{U}_{\mathbb{P}}(x(0),J_\mathrm{B}(0))$, which implies that $x(0) \in \mathcal{X}_{\mathbb{P}}(J_\mathrm{B}(0))$, and hence feasibility of \eqref{eq:stability_filter} at time $k=0$. For the induction step, assume feasibility at some time $k \in \mathbb{I}_{\geq 0}$. Feasibility at time $k+1$ follows from Assumption \ref{ass:evol_stab_bound}, ensuring existence of a feasible candidate solution $u^\mathrm{c}_{\cdot|k+1} \in \mathcal{U}_{\mathbb{P}}(x(k+1),J_\mathrm{B}(k+1))$ for the successor time step.
    
    \textit{Bounded convergence:}
    For stability in terms of bounded convergence, we use Lyapunov-like arguments in the following. With Assumptions \ref{ass:evol_stab_bound}- \ref{ass:init_stab_bound} and recursive feasibility, for any $x(0) \in \mathcal{X}$ and all $k \in \mathbb{I}_{\geq 0}$ we have that  
  \begin{equation}\label{eq:upper_bound}
      V(x(k),k) \leq V(x(0),0) \leq J_\mathrm{B}(0) \leq \alpha_2(\lVert x(0) \rVert).
  \end{equation}
    By Assumption \ref{ass:cont_stagecost}, definition of $J(x(k),u_{\cdot|k})$, and with recursive feasibility, it holds for all $k \in \mathbb{I}_{\geq 0}$ that 
    \begin{equation}\label{eq:lower_bound}
        V(x(k),k) \geq \alpha_{\ell}(\lVert x(k) \rVert) =: \alpha_1(\lVert x(k) \rVert).
    \end{equation}
     Furthermore, from Assumptions \ref{ass:cont_stagecost}-\ref{ass:evol_stab_bound} and recursive feasibility, for all $k \in \mathbb{I}_{\geq 0}$, we get that
    \begin{multline}\label{eq:nom_stab_filter_proof_decr}
        V(x(k+1),k+1)-V(x(k),k) \\ \leq -\zeta_\mathrm{min} \alpha_{\ell}(\lVert x(k) \rVert) =: -\alpha_3(\lVert x(k) \rVert).
    \end{multline}
    For any $\varepsilon>0$ it follows from $\mathcal{X}$ containing a neighborhood of the origin that there exists $\eta_1$ with $0<\eta_1\leq\varepsilon$ such that $\mathcal{B}_{\eta_1}(x):=\{x \in \mathbb{R}^n \mid \lVert x \rVert \leq \eta_1 \} \subseteq \mathcal{X}$. Also, since $\alpha_1, \alpha_2 \in \mathcal{K}_\infty$, for any $\varepsilon_1$ with $0< \varepsilon_1 \leq \eta_1$, there exists a $\delta_1 \in (0, \varepsilon_1)$ such that $\alpha_2(\delta_1) < \alpha_1(\varepsilon_1)$. This is true since $\alpha_2(\delta_1) \rightarrow 0$ as $\delta_1 \rightarrow 0$. Based on this, with recursive feasibility, and with \eqref{eq:upper_bound} - \eqref{eq:lower_bound}, for any $x(0) \in \mathcal{X}$ with $\lVert x(0) \rVert < \delta_1$ and all $k \in \mathbb{I}_{\geq 0}$, we have that
    \begin{multline}
        \alpha_1(\lVert x(k) \rVert) \leq V(x(k),k) \leq \alpha_2(\lVert x(0) \rVert) \\ < \alpha_2(\delta_1) < \alpha_1(\varepsilon_1).
    \end{multline}
    Hence, for all $k \in \mathbb{I}_{\geq 0}$, it holds that $\alpha_1(\lVert x(k) \rVert) < \alpha_1(\varepsilon_1)$, which, with monotonicity of $\alpha_1$, implies that $ \lVert x(k) \rVert < \varepsilon_1$ holds for all $k \in \mathbb{I}_{\geq 0}$. For convergence, we define $\Delta V(k) = V(x(k+1),k+1) - V(x(k),k) \leq 0$. Since, with recursive feasibility and Assumption \ref{ass:evol_stab_bound}, we have that $\Delta V(k) \leq 0$ for all $k \in \mathbb{I}_{\geq 0}$, implying that $V(x(k),k)$ monotonically decreases over time. Furthermore, $J(\cdot)$ is lower bounded by zero. Hence, it follows that $\lim_{k \rightarrow \infty} V(x(k),k) = V_\mathrm{L}$ exists. Then, this implies that $\lim_{k \rightarrow \infty} \Delta V(k) = V_\mathrm{L}-V_\mathrm{L}=0$. Since, with \eqref{eq:nom_stab_filter_proof_decr}, it holds that $0 \leq \alpha_3(\lVert x_{0|k} \rVert) \leq -\Delta V(k)$, it follows that $\lim_{k \rightarrow \infty} \alpha_3(\lVert x(k) \rVert) = 0$, which implies that $\lVert x(k) \rVert \rightarrow 0$ for $k \rightarrow \infty$ and hence concludes the proof.
\end{proof}
    With the formal requirements and theoretical derivation in place, we discuss possible design choices regarding the ingredients of \eqref{eq:stability_filter} satisfying the necessary assumptions. 

    \textbf{Design of $\ell$, $V_\mathrm{f}$, and $\mathbb{X}_\mathrm{f}$:}
    We leverage standard MPC results to ensure existence of some $u^\mathrm{c}_{\cdot|k+1}$ and $\rho \in (0,1]$ satisfying Assumption \ref{ass:evol_stab_bound} for all $k \in \mathbb{I}_{\geq 0}$, where we distinguish between explicit terminal ingredients and local controllability assumptions. 
    
    \textit{With terminal ingredients}: The standard MPC design \cite{MPC_rawlings} relies on computing a terminal state feedback $\kappa: \mathbb{R}^n \rightarrow \mathbb{R}^m$, a terminal set $\mathbb{X}_{\mathrm{f}}$, and a terminal cost $V_\mathrm{f}$, such that $\kappa$ renders $\mathbb{X}_\mathrm{f}$ positive invariant and such that $V_\mathrm{f}$ admits a control Lyapunov function.  With this, at time $k$ given $x(k) \in \mathcal{X}_\mathbb{P}(J_\mathrm{B}(k))$ and an input sequence $u_{\cdot|k} \in \mathcal{U}_\mathbb{P}(x(k),J_\mathrm{B}(k))$, the candidate 
    \begin{equation}
        u_{i | k+1}^\mathrm{c}~:=~\begin{cases} 
            u_{i+1|k}, & i \in \mathbb{I}_{[0,N-2]} \\
            \kappa(x_{N|k}) & i = N-1
            \end{cases}
    \end{equation}
    implies $\rho = 1$ \cite{MPC_rawlings}.
    A special case of this would be to choose $V_\mathrm{f}(x):=0$, $\mathbb{X}_\mathrm{f}:= \{0\}$ and $\kappa(x):=0$, which may result in a smaller set of feasible states, but may be attractive for practical problems where designing more advanced terminal ingredients is difficult.
            
    \textit{Without terminal ingredients}: Existence of some $u_{\cdot | k+1}^\mathrm{c}$ satisfying Assumption \ref{ass:evol_stab_bound} with some \textit{suboptimality index} $\rho \in (0,1]$ can also be ensured by choosing $V_\mathrm{f}(x) = 0$ and $\mathbb{X}_\mathrm{f}= \mathbb{R}^n$ for some sufficiently long prediction horizon $N$. We note that for this approach a so-called local controllability condition and additional requirements on $J(\cdot)$ and $\ell(x,u)$ are needed, compare with \cite{unconstr_gruene}.

    It is important to note that for any of the described design choices, for all $k \in \mathbb{I}_{\geq 0}$ and $x(k) \in \mathcal{X}$, it holds that $V^\mathrm{MPC}(x(k)) \leq \alpha_2(\lVert x(k) \rVert)$ for some $\alpha_2 \in \mathcal{K}_\infty$ \cite[chapter 2]{MPC_rawlings}, which paves the way for Assumption \ref{ass:init_stab_bound}.

    \textbf{Design of $J_\mathrm{B}(k)$:} A simple design choice of $J_\mathrm{B}(k)$ is given by initially solving a corresponding MPC problem to obtain $V^\mathrm{MPC}(x(0))$, setting $J_\mathrm{B}(0) = \gamma V^\mathrm{MPC}(x(0))$ for some user-defined $\gamma \geq 1$. The evolution of $J_\mathrm{B}(k)$ for all $k \in \mathbb{I}_{>0}$ can then be selected as
    \begin{equation}\label{eq:exmpl_stab_bound_convergence}
        J_\mathrm{B}(k) := V(x(k-1),k-1) - \zeta(k) \ell(x(k-1),u^*_{0|k-1}).
    \end{equation}
    With this, assuming $\ell$, $V_\mathrm{f}$, and $\mathbb{X}_\mathrm{f}$ are designed, e.g., as described above, Assumption \ref{ass:init_stab_bound} is trivially satisfied. Additionally, Assumption \ref{ass:evol_stab_bound} is satisfied for any (user-defined) $\zeta(k) \in [\zeta_{\mathrm{min}},\rho]$.

\subsection{Uniform asymptotic stability}\label{sec:stability}
To further strengthen the stability properties from bounded convergence to uniform asymptotic stability, the central approach is to impose stricter requirements on the stability bound. This enables us to apply standard Lyapunov arguments on $V(x(k),k)$. In particular, instead of just requiring that $J_\mathrm{B}(0)$ is bounded with respect to the initial state $x(0)$, we now require an upper bound on $J_\mathrm{B}(k)$ that depends on the current state $x(k)$ for all $k\in \mathbb{I}_{\geq 0}$. 
\begin{ass}\label{ass:uniform_stability_bound}
     There exists an $\alpha_2\in\mathcal{K}_\infty$ such that for any $k \in \mathbb{I}_{\geq 0}$ and $x(k) \in \mathcal{X}$, the stability bound $J_\mathrm{B}(k)$ satisfies
 \begin{equation}
     V^\mathrm{MPC}(x(k)) \leq J_\mathrm{B}(k) \leq \alpha_2(\lVert x(k) \rVert).
 \end{equation}
\end{ass}

 With this assumption in place, we employ standard Lyapunov arguments using the time-varying Lyapunov function candidate $V(x(k),k)$ to show uniform asymptotic stability by extending the findings from Section \ref{sec:convergence}, compare with \cite[chapter B]{MPC_rawlings}.

\begin{theo}\label{theo:unif_stab}
    Let Assumptions \ref{ass:cont_stagecost} - \ref{ass:evol_stab_bound} and \ref{ass:uniform_stability_bound} hold and suppose that $\mathcal{X}$ contains a neighborhood of the origin. Then application of Algorithm~\ref{algo:nom_stabfilter_nom} with $\mathbb{P} := \mathbb{P}(k)$ yields uniform asymptotic stability according to Definition~\ref{defi:stability},~(2) with respect to $x^*~=~0$ for any $x(k_0) \in \mathcal{X}$ with $k_0 \in \mathbb{I}_{\geq 0}$ and all $u_\mathrm{des}(k) \in \mathbb{R}^m$ with $k \in \mathbb{I}_{\geq k_0}$.
\end{theo}
\begin{proof}
   Since satisfaction of Assumption \ref{ass:uniform_stability_bound} implies satisfaction of Assumption \ref{ass:init_stab_bound}, recursive feasibility follows from Theorem~\ref{theo:general_stab_filter}. For uniform asymptotic stability, using similar arguments as in the proof of Theorem \ref{theo:general_stab_filter}, for any $x(k) \in \mathcal{X}$ and all $k \in \mathbb{I}_{\geq k_0}$ with $k_0 \in \mathbb{I}_{\geq 0}$, one obtains
    \begin{equation}\label{eq:stability_proof_bounds}
        \alpha_1(\lVert x(k) \rVert ) \leq V(x(k),k) \leq \alpha_2(\lVert x(k) \rVert),
    \end{equation}
    \begin{equation}\label{eq:stability_proof_decr}
        V(x(k+1),k+1) - V(x(k),k) \leq - \alpha_3(\lVert x(k) \rVert).
    \end{equation}
    Since $\mathcal{X}$ contains a neighborhood of the origin, for any $\varepsilon>0$, there exists $\eta_2$ such that $\mathcal{B}_{\eta_2}(x) := \{x \in \mathbb{R}^n \mid \lVert x \rVert \leq \eta_2 \} \subseteq \mathcal{X}$. Again, since $\alpha_1,\alpha_2 \in \mathcal{K}_\infty$, for any $0 < \varepsilon_2 \leq \eta_2$ there exists $\delta_2 \in (0,\varepsilon_2)$ such that $\alpha_2(\delta_2) < \alpha_1(\varepsilon_2)$. By Assumption~\ref{ass:uniform_stability_bound}, for any $x(k_0) \in \mathcal{X}$, we have that $x(k_0) \in \mathcal{X}_\mathbb{P}(J_\mathrm{B}(k))$. Hence, for any $x(k_0) \in \mathcal{X}$ with $\lVert x(k_0) \rVert < \delta_2$ and $k_0 \in \mathbb{I}_{\geq 0}$, all $k \in \mathbb{I}_{\geq k_0}$, with Assumption \ref{ass:uniform_stability_bound}, recursive feasibility, and \eqref{eq:stability_proof_bounds} - \eqref{eq:stability_proof_decr}, we have that
    \begin{multline}
        \alpha_1(\lVert x(k) \rVert ) \leq V(x(k),k) \leq V(x(k_0),k_0) \\ \leq \alpha_2(\lVert x(k_0) \rVert ) < \alpha_2(\delta_2) < \alpha_1(\varepsilon_2).
    \end{multline}
    Thus, $\alpha_1(\lVert x(k) \rVert) < \alpha_1(\varepsilon_2)$ for all $k \in \mathbb{I}_{\geq 0}$, which, with monotonicity of $\alpha_1$, implies that $\lVert x(k) \rVert < \varepsilon_2$ for all $k~\in~\mathbb{I}_{\geq k_0}$. Finally, convergence can be concluded using the arguments of the proof of Theorem \ref{theo:general_stab_filter}, which concludes the proof. 
\end{proof}

In the following, we discuss possible design choices regarding the ingredients of \eqref{eq:stability_filter}, leveraging results from suboptimal MPC~\cite{allan} to satisfy Assumptions~\ref{ass:evol_stab_bound} and \ref{ass:uniform_stability_bound}. 

\textbf{Design of $\ell$, $V_\mathrm{f}$, $\mathbb{X}_\mathrm{f}$:}
 Similar to \cite{didier2024predictive}, the main idea is to adopt the setting considered in \cite{allan} to construct suboptimal \textit{warm start} solutions, which are used to compute the stability bound. To this end, as in \cite{allan}, we assume that $\mathbb{X}$ is closed, and require the following assumptions on the cost terms and terminal ingredients. 
\begin{ass}\label{ass:stability_ass} 
The terminal set $\mathbb{X}_\mathrm{f} \subseteq \mathbb{X}$ is closed and contains the origin in its interior. Furthermore, there exists a terminal control law $\kappa_\mathrm{f}:\mathbb{X}_\mathrm{f} \rightarrow \mathbb{U} $ such that for all $x \in \mathbb{X}_\mathrm{f}$:
\begin{itemize}
    \item [(i)] $\mathbb{X}_\mathrm{f}$ is positively invariant with respect to $x(k+1) = f(x(k),\kappa_\mathrm{f}(x(k)))$,
    \item[(ii)] $V_\mathrm{f}(f(x,\kappa_\mathrm{f}(x)))- V_\mathrm{f}(x) \leq -\ell(x,\kappa_\mathrm{f}(x))$.
\end{itemize}
\end{ass}

We denote the set of \textit{admissible warm starts}  $\Tilde{u}_{\cdot|k}$ by
\begin{multline}\label{eq:U_tilde}
     \Tilde{\mathcal{U}}(x(k)) := \{\Tilde{u}_{\cdot | k} \in \mathcal{U}(x(k)) \mid \\J(x(k),\Tilde{u}_{\cdot | k}) \leq V_\mathrm{f}(x(k)) \text{ if } x(k) \in \mathbb{X}_\mathrm{f} \},
 \end{multline}
 and all \textit{admissible state and warm start pairs} $(x(k),\Tilde{u}_{\cdot | k})$ by
 \begin{equation}\label{eq:Z_tilde}
     \Tilde{\mathcal{Z}}:= \{(x(k),\Tilde{u}_{\cdot | k}) \mid x(k) \in \mathcal{X} \text{ and } \Tilde{u}_{\cdot |k} \in \Tilde{\mathcal{U}}({x(k))} \}.
 \end{equation}

Note that if $x(k) \in \mathbb{X}_\mathrm{f}$, we can always construct an admissible warm start sequence $\Tilde{u}_{\cdot|k} \in \Tilde{\mathcal{U}}(x(k))$ using the terminal control law $\kappa_\mathrm{f}$, i.e., $\Tilde{u}^\mathrm{f}_{\cdot|k}(x(k)):= (\kappa_\mathrm{f}(x(k)), \kappa_\mathrm{f}(f(x(k),\kappa_\mathrm{f}(x(k)))), \ldots)$ belongs to $\Tilde{\mathcal{U}}(x(k))$ \cite{allan}. Together with Assumption \ref{ass:stability_ass}, this implies for all $x(k+1) = f(x(k),u_{0|k})$ with $x(k) \in \mathcal{X}_\mathbb{P}(J_\mathrm{B}(k))$ and $u_{\cdot|k} \in  \mathcal{U}_\mathbb{P}(x(k),J_\mathrm{B}(k))$ that there exists a candidate input sequence $u^\mathrm{c}_{\cdot|k+1} \in \Tilde{\mathcal{U}}(x(k+1))$ satisfying Assumption \ref{ass:evol_stab_bound} with $\rho = 1$ \cite{allan}.

\textbf{Design of $J_\mathrm{B}(k)$:}
To establish a stability bound satisfying Assumptions~\ref{ass:evol_stab_bound} and \ref{ass:uniform_stability_bound} without solving an MPC problem to determine $J_\mathrm{B}(k)$ with $V^\mathrm{MPC}(x(k))$, we further leverage the arguments from~\cite{allan}. In the following, we start with establishing Assumption \ref{ass:uniform_stability_bound}, showing that $J(x(k),\Tilde{u}_{\cdot|k}) \leq \alpha_2(\lVert x(k) \lVert)$ with $\Tilde{u}_{\cdot|k} \in \Tilde{\mathcal{U}}(x(k))$ holds for all $k \in \mathbb{I}_{\geq 0}$. With this we satisfy the upper bound in Assumption \ref{ass:uniform_stability_bound} by defining the stability bound as $J_\mathrm{B}(k) = J(x(k),\Tilde{u}_{\cdot|k})$. At the same time, with $V^\mathrm{MPC}(x(k)) \leq J(x(k),\Tilde{u}_{\cdot|k})$ holding for all $k \in \mathbb{I}_{\geq 0}$ due to suboptimality, we enable a possibly improved input matching. To this end, we first collect existing results available, e.g., in \cite{allan}, which we leverage afterwards to establish Assumption \ref{ass:uniform_stability_bound}. 
\begin{prop}\label{prop:results_subopt}
     Let Assumptions \ref{ass:cont_stagecost} and \ref{ass:stability_ass} hold. Then:
    \begin{itemize}
        \item[(1)]  For all $(x(k),\Tilde{u}_{\cdot|k}) \in \Tilde{\mathcal{Z}}$ and $k \in \mathbb{I}_{\geq 0}$, there exists a $\mathcal{K}_\infty$ function $\alpha_5$ such that $\lVert \Tilde{u}_{\cdot|k} \rVert \leq \alpha_5(\lVert x(k) \rVert)$ \cite[proposition 10]{allan}.

        \item[(2)] For all $(x(k),u_{\cdot|k}) \in \mathcal{Z}$ and $k \in \mathbb{I}_{\geq 0}$, there exists a $\mathcal{K}_\infty$ function $\alpha_4$ such that ~\cite[proposition 7]{allan}:
        \begin{equation}\label{eq:warm_start_upper_bound_eq1}
        J(x(k),u_{\cdot|k}) \leq \alpha_4(\lVert (x(k),u_{\cdot|k}) \rVert) \leq \alpha_4(\lVert (x(k)\rVert + \lVert u_{\cdot|k} \rVert).
    \end{equation}
    \end{itemize}
\end{prop}

With these results, the desired upper bound on $J(x(k),\Tilde{u}_{\cdot|k})$ can be concluded, paving the way for Assumption \ref{ass:uniform_stability_bound}.
\begin{prop}\label{prop:unif_upper_bound}
    Let Assumptions \ref{ass:cont_stagecost} and \ref{ass:stability_ass} hold. Then, for all $(x(k),\Tilde{u}_{\cdot|k}) \in \Tilde{\mathcal{Z}}$ and all $k \in \mathbb{I}_{\geq 0}$, there exists a function $\alpha_2(\cdot) \in \mathcal{K}_\infty$ such that $J(x(k),\Tilde{u}_{\cdot|k}) \leq \alpha_2(\lVert x(k) \rVert)$.
\end{prop}
\begin{proof}
    With Proposition \ref{prop:results_subopt} and since $\Tilde{\mathcal{Z}} \subseteq \mathcal{Z}$ by definition, for all $(x(k),\Tilde{u}_{\cdot|k}) \in \Tilde{\mathcal{Z}}$ and $k \in \mathbb{I}_{\geq 0}$, this implies
    \begin{multline}
            J(x(k),\Tilde{u}_{\cdot|k}) \leq \alpha_4(\lVert x(k) \rVert + \lVert \Tilde{u}_{\cdot|k} \rVert) \\ \leq   \alpha_4(\lVert x(k) \rVert + \alpha_5(\lVert x(k)\rVert)) =: \alpha_2(\lVert x(k)\rVert),
    \end{multline}
    with $\alpha_2 \in \mathcal{K}_\infty$, which concludes the proof.
\end{proof}

Satisfaction of Assumption \ref{ass:uniform_stability_bound} is achieved by choosing $J_\mathrm{B}(k) = J(x(k),\Tilde{u}_{\cdot|k})$ with $\Tilde{u}_{\cdot|k} \in \Tilde{\mathcal{U}}(x(k))$. To satisfy Assumption~\ref{ass:evol_stab_bound}, we additionally require the stability bound to enforce the necessary decrease of the closed-loop stability cost $V(x(k),k)$. Based on this, we denote the set of \textit{stabilizing successor warm start sequences} by 
 \begin{multline}
     \hat{\mathcal{U}}(x(k),u^*_{\cdot|k})  := \{\Tilde{u}_{\cdot|k+1} \in \Tilde{\mathcal{U}}(x(k+1)) \mid J(x(k+1), \Tilde{u}_{\cdot|k+1}) \\ \leq V(x(k),k)-\ell(x(k),u^*_{0|k}) \}.
 \end{multline}

With this in place, we define $J_\mathrm{B}(k)$ at time $k$ by 
\begin{equation}\label{eq:stab_bound_subopt}
J_\mathrm{B}(k) := J(x(k),\Tilde{u}_{\cdot|k}),     
\end{equation}
requiring $\Tilde{u}_{\cdot|k} \in \Tilde{\mathcal{U}}(x(k))$ for all $k \in \mathbb{I}_{\geq 0}$ and $\Tilde{u}_{\cdot|k} \in \hat{\mathcal{U}}(x(k-1),u^*_{\cdot|k-1})$ for all $k \in \mathbb{I}_{>0}$.
As a result, Assumptions~\ref{ass:evol_stab_bound} and \ref{ass:uniform_stability_bound} are satisfied, where Assumption \ref{ass:evol_stab_bound} holds with $\rho = 1$ for any $\zeta(k+1) \in [\zeta_{\mathrm{min}},1]$. Additionally, by noting that Assumption \eqref{ass:stability_ass} requires $V_\mathrm{f}$ to be a \textit{local} CLF on $\mathbb{X}_\mathrm{f}$, a conceptual connection to CBF approaches~\cite{CBF_Agrawal} arises. For the special case $N=1$, \eqref{eq:stability_constraint} embodies a classical CLF constraint where the size of $\mathcal{X}_\mathbb{P}$ is strongly determined through $\mathbb{X}_\mathrm{f}$. For a prediction horizon $N>1$, $V(x(k),k)$ is implicitly defined through the optimal solution to \eqref{eq:stability_filter} and \eqref{eq:stability_constraint} can be interpreted as an implicit CLF constraint.

 \begin{rem}
    The proposed framework can be enhanced by robust techniques such as tube-based MPC, compare, e.g., with \cite[chapter 3]{MPC_rawlings}. For example, a robust tube can be built around the nominal (uniformly asymptotically) stable trajectory resulting from the application of Algorithm \ref{algo:nom_stabfilter_nom}. With this, robust (uniform asymptotic) stability as well as robust constraint satisfaction can be guaranteed, applying standard constraint tightening techniques \cite[chapter 3]{MPC_rawlings}.
\end{rem}

\subsection{Performance}
In addition to safety and the desired stability specification, the framework provides the following implicit closed-loop performance bound. For all $T \in \mathbb{I}_{> 0}$ it holds that
\begin{equation}\label{eq:perfo_bound_convergence}
    V(x(T),T) - V(x(0),0) \leq - \sum_{i=0}^{T-1} \zeta(i+1)\ell(x(i),u_{0|i}^*),
\end{equation}
which can be shown by substituting \eqref{eq:evolu_stab_bound} into \eqref{eq:stability_constraint} and summing up from $k=0$ to $k=T$. The fact that we only require $\zeta(k)$ to satisfy $\zeta(k) \in [\zeta_\mathrm{min},\rho]$, but not require it to be constant over time, enables us to mediate improved input matching and performance on-the-fly depending on the environment or specific needs. This can be achieved by choosing $\zeta(k) = \zeta_{\mathrm{min}}$, or $\zeta(k) = \rho$, respectively. This could be especially helpful when injecting, e.g., learning-based or randomized inputs to systems for which it is important to maintain a certain performance level. It is also worth mentioning that this result holds for general, non-stabilizing stability costs. 

\section{Extension to Trajectory Stabilization}\label{sec:Extensions}
The proposed framework can be extended to the stabilization of dynamic reference signals $r(k)=(x^\mathrm{r}(k),u^\mathrm{r}(k)) \in \mathbb{R}^{n+m}$ using results presented in \cite{Koehler_reference_gen}. Such references are, e.g., provided by a higher-level planner in case of autonomous driving, or pre-computed trajectories connecting different setpoints as commonly employed in process industry. In the following, we assume that the reference $r(k)$ satisfies $r(k) \in Z^\mathrm{r}$ with $ Z^\mathrm{r}\subseteq \mathrm{int}(\mathbb{X}) \times \mathrm{int}(\mathbb{U})$ and $r(k+1) \in E(r(k))$ with $E(r(k)):=\{r(k+1) \in Z^\mathrm{r} \mid x^\mathrm{r}(k+1)=f(x^\mathrm{r}(k),u^\mathrm{r}(k))\}$ for all $k \in \mathbb{I}_{\geq 0}$. With this, we ensure that the reference is feasible and reachable. 

The goal is that system~\eqref{eq:nonlinear_system} stabilizes $r(k)$ with minimal input modification. More precisely, we aim at filtering desired inputs regarding safety and stability, considering now stability with respect to the tracking error $e^\mathrm{r} = x-x^\mathrm{r}=0$ in terms of the following generalization of Definition \ref{defi:stability},~(1).
\begin{defi}\label{defi:stability_tracking}
     Consider system \eqref{eq:nonlinear_system} subject to a time-varying feedback law $u(k) = \mu(x(k),k)$ and let $r(k)=(x^\mathrm{r}(k),u^\mathrm{r}(k)) \in Z^\mathrm{r}$ be a reference signal for system \eqref{eq:nonlinear_system}. We call the resulting closed-loop system \textit{bounded and converging} to the tracking error $e^\mathrm{r}=x-x^\mathrm{r}=0$ if, for each $\varepsilon > 0$, there exists a $\delta = \delta(\varepsilon)>0$ such that $\lVert e^\mathrm{r}(0) \rVert < \delta \Rightarrow \lVert e^\mathrm{r}(k) \rVert < \varepsilon$ for all $k \in \mathbb{I}_{\geq 0}$, and $ e^\mathrm{r}(k) \rightarrow 0$ as $k \rightarrow \infty$.
\end{defi}

Let the reference over the prediction horizon $N$ be denoted by $r_{\cdot|k} = \{r_{i|k}\}_{i=0}^{N}$ with $r_{i|k} = r(i+k), i \in \mathbb{I}_{[0,N]}$. Similar to Section \ref{sec:MPSF}, the idea is to constrain the tracking stability cost
\begin{equation}\label{eq:stab_cost_traj_track}
        J^\mathrm{r}(x(k),u_{\cdot|k},r_{\cdot|k}) = \sum_{i=0}^{N-1}\ell^\mathrm{r}(x_{i|k},u_{i|k},r_{i|k}) + V^\mathrm{r}_\mathrm{f}(x_{N|k},x^\mathrm{r}_{N|k}),
\end{equation}
consisting of a tracking stage and terminal cost $\ell^\mathrm{r}$ and $V_\mathrm{f}^\mathrm{r}$, with a tracking stability bound $J^\mathrm{r}_{\mathrm{B}}(k)$ to enforce $J^\mathrm{r}(x(k),u_{\cdot|k},r_{\cdot|k}) \leq J^\mathrm{r}_{\mathrm{B}}(k)$. 
\begin{ass}\label{ass:tracking_cost_terms}
    The tracking stage and terminal cost $\ell^\mathrm{r}:\mathbb{R}^n \times \mathbb{R}^m \times \mathbb{R}^{n+m} \rightarrow \mathbb{R}_{\geq 0}$ and $V^\mathrm{r}_\mathrm{f}:\mathbb{R}^n \times \mathbb{R}^n \rightarrow \mathbb{R}_{\geq 0}$ are continuous and satisfy $\ell^\mathrm{r}(x,u,(x,u))=V_\mathrm{f}^\mathrm{r}(x,x) = 0$. Furthermore, there exists an $\alpha_\ell \in \mathcal{K}_\infty$ such that $\ell^\mathrm{r}(x,u,(x^\mathrm{r},u^\mathrm{r})) \geq \alpha_\ell(\lVert x - x^\mathrm{r} \rVert)$ for all $(x,u) \in \mathbb{X} \times \mathbb{U}$ and $r \in Z^\mathrm{r}$.
\end{ass} 

We construct the \textit{stability-enhanced tracking safety filter problem} $\mathbb{P}^\mathrm{r}(k)$ similar to problem \eqref{eq:stability_filter} by
 \begin{subequations}\label{eq:stability_filter_tracking}
\begin{align}
    \min_{u_{\cdot|k}} \ & \eqref{eq:safe_filter_cost}\\
    \text{s.t. } \  & \eqref{eq:safety_filter_init_constr} - \eqref{eq:safety_filter_input_constr} \label{eq:tracking_stab_filter_safety_filter_parts}\\
 & x_{N|k} \in \mathbb{X}^\mathrm{r}_\mathrm{f}(r_{N|k}),\label{eq:changing_setp_term_constr_}\\ 
 & J^\mathrm{r}(x(k),u_{\cdot|k},r_{\cdot|k}) \leq J^\mathrm{r}_{\mathrm{B}}(k)\label{eq:changing_setp_stability_constr},
\end{align}
\end{subequations}
substituting \eqref{eq:safety_filter_term_constr} and \eqref{eq:stability_constraint} with \eqref{eq:changing_setp_term_constr_} and \eqref{eq:changing_setp_stability_constr}, respectively.
We define the set of feasible input sequences $u_{\cdot|k}$ for \eqref{eq:stability_filter_tracking} by
\begin{equation}
    \mathcal{U}_{\mathbb{P}^\mathrm{r}}(x(k),J^\mathrm{r}_\mathrm{B}(k),r_{\cdot|k}):=\{u_{\cdot|k} \mid \eqref{eq:tracking_stab_filter_safety_filter_parts} -\eqref{eq:changing_setp_stability_constr} \},
\end{equation} and the set of feasible states $x(k)$ by
\begin{equation}\label{eq:stability_filter_feasible_set}
    \mathcal{X}_{\mathbb{P}^\mathrm{r}}(J_\mathrm{B}^\mathrm{r}(k),r_{\cdot|k}):=\{x(k) \mid \exists u_{\cdot|k} \in \mathcal{U}_{\mathbb{P}^\mathrm{r}}(x(k),J^\mathrm{r}_\mathrm{B}(k),r_{\cdot|k})\}.
\end{equation}
 Furthermore, we define 
\begin{equation}
    \mathcal{U}^\mathrm{r}(x(k),r_{\cdot|k}):= \{ u_{\cdot|k} \mid \eqref{eq:safety_filter_init_constr} - \eqref{eq:safety_filter_input_constr}, \eqref{eq:changing_setp_term_constr_} \}
\end{equation}
and 
\begin{equation}
    \mathcal{X}^\mathrm{r}(r_{\cdot|k}):= \{ x(k) \mid \exists u_{\cdot|k} \in  \mathcal{U}^\mathrm{r}(x(k),k) \},
\end{equation}
which can be considered as the tracking equivalents of $\mathcal{U}(x(k))$ and $\mathcal{X}$, respectively.
The solution to \eqref{eq:stability_filter_tracking} is the optimal input sequence $u^{\mathrm{r},*}_{\cdot|k}$ and the closed-loop tracking stability cost 
\begin{equation}
    V^\mathrm{r}(x(k),k):= J^\mathrm{r}(x(k),u^{\mathrm{r},*}_{\cdot|k},r_{\cdot|k}).
\end{equation}
Similar as in the previous sections, the input applied to system \eqref{eq:nonlinear_system} is defined by the first element of the optimal solution, i.e., $u(k):=u^{\mathrm{r},*}_{0|k}$.

 In the following, we establish the desired safety and stability guarantees by adapting Assumptions \ref{ass:evol_stab_bound}-\ref{ass:init_stab_bound} to the tracking case as follows.

\begin{ass}\label{ass:evol_stab_bound_tracking}
 Consider $\zeta_{\mathrm{min}} >0$ and let $\rho \in [\zeta_{\mathrm{min}},1]$. We assume that, for all $k \in \mathbb{I}_{\geq 0}$, $x(k) \in \mathcal{X}_{\mathbb{P}^\mathrm{r}}(J_{\mathrm{B}}^\mathrm{r}(k),r_{\cdot|k})$, and any $u^{\mathrm{r},*}_{\cdot|k} \in \mathcal{U}_{\mathbb{P}^\mathrm{r}}(x(k),J_\mathrm{B}^\mathrm{r}(k),r_{\cdot|k})$, there exists a candidate input sequence $u^\mathrm{c}_{\cdot|k+1} \in \mathcal{U}^\mathrm{r}(x(k),r_{\cdot|k})$ and $\zeta(k+1) \in [\zeta_{\mathrm{min}},\rho]$ such that 
\begin{multline}\label{eq:nom_stabfilter_evolu_stab_bound}
  J^\mathrm{r}(x(k+1),u^\mathrm{c}_{\cdot|k+1},r_{\cdot|k+1}) \leq V^\mathrm{r}(x(k),k)-\rho \ell(x(k),u^{\mathrm{r},*}_{0|k},r_{0|k}) \\ \leq  J^\mathrm{r}_\mathrm{B}(k+1)  \leq V^\mathrm{r}(x(k),k)  -\zeta(k+1) \ell^\mathrm{r}(x(k),u^{\mathrm{r},*}_{0|k},r_{0|k}).
\end{multline}
\end{ass}
We again define a lower bound using the optimal value function of the \textit{corresponding} (tracking) MPC problem, which is given by
\begin{equation}
    V^{\mathrm{r, MPC}}(x(k),r_{\cdot|k}):= \min_{u_{\cdot|k} \in \mathcal{U}^\mathrm{r}(x(k),r_{\cdot|k})} J^\mathrm{r}(x(k),u_{\cdot|k},r_{\cdot|k}).
\end{equation}
\begin{ass}\label{ass:init_stab_bound_traj_track}
 There exists an $\alpha_2 \in \mathcal{K}_\infty$ such that for any $x(0) \in \mathcal{X}^\mathrm{r}(r_{\cdot|0})$, the initial stability bound $J^\mathrm{r}_\mathrm{B}(0)$ satisfies
 \begin{equation}
     V^\mathrm{r, MPC}(x(0),r_{\cdot|0}) \leq J^\mathrm{r}_\mathrm{B}(0) \leq \alpha_2(\lVert x(0) - x^\mathrm{r}(0) \rVert).
 \end{equation}
\end{ass}

Finally, this allows us to establish stability according to Definition \ref{defi:stability_tracking}.
\begin{theo}\label{theo:tracking_stability}
  Let Assumptions \ref{ass:tracking_cost_terms} - \ref{ass:init_stab_bound_traj_track} hold and suppose that $\mathcal{X}^\mathrm{r}(r_{\cdot|0})$ contains a neighborhood of $x^\mathrm{r}(0)$. Then application of Algorithm~\ref{algo:nom_stabfilter_nom} with $\mathbb{P}:=\mathbb{P}^\mathrm{r}(k)$ yields bounded convergence according to Definition \ref{defi:stability_tracking} with respect to the tracking error $e^\mathrm{r}=0$ for any $x(0) \in \mathcal{X}^\mathrm{r}(r_{\cdot|0})$ and all $u_\mathrm{des}(k) \in \mathbb{R}^m$ with $k \in \mathbb{I}_{\geq 0}$. 
\end{theo}
\begin{proof}
    Similar as before, initial feasibility as well as following recursive feasibility results from Assumptions~\ref{ass:evol_stab_bound_tracking} and \ref{ass:init_stab_bound_traj_track}. For stability in terms of bounded convergence, using similar arguments as in the proof of Theorem \ref{theo:general_stab_filter}, for any $r(0) \in Z^\mathrm{r}$, any $x(0) \in \mathcal{X}^\mathrm{r}(r_{\cdot|0})$, and all $k \in \mathbb{I}_{\geq 0}$, we have that
    \begin{equation}\label{eq:tracking_proof_lowerb}
        \alpha_1(\lVert x(k) - x^\mathrm{r}(k) \rVert)  \leq V^\mathrm{r}(x(k),k)  \leq \alpha_2(\lVert x(0) - x^\mathrm{r}(0) \rVert),
    \end{equation}
      \begin{equation}\label{eq:tracking_proof_decr}
            V^\mathrm{r}(x(k+1),k+1) - V^\mathrm{r}(x(k),k)  \\ \leq  \alpha_3(\lVert x(k) - x^\mathrm{r}(k) \rVert).  
    \end{equation}
    Based on this, using similar arguments as in the proof of Theorem \ref{theo:general_stab_filter}, \eqref{eq:tracking_proof_lowerb}-\eqref{eq:tracking_proof_decr} imply that $\lVert x(k) - x^\mathrm{r}(k) \rVert = \lVert e(k) \rVert < \varepsilon_3$ holds for all $k \in \mathbb{I}_{\geq 0}$ as well as that $\lVert e(k) \rVert \rightarrow \infty$ as $k \rightarrow \infty$.
\end{proof}
It should be noted that the tracking framework retains the performance bound 
\begin{equation}\label{eq:tracking_perfo_bound}
    V^\mathrm{r}(x(T),T)-V^\mathrm{r}(x(0),0)  \leq - \sum_{i=0}^{T-1} \zeta(i+1) \ell^\mathrm{r}(x(i),u^{\mathrm{r},*}_{0|i},r_{0|i})
\end{equation}
with the tuning parameter $\zeta(k+1)$ holding for all $T \in \mathbb{I}_{>0}$, allowing to interpolate between input matching and 'stability' during online operation. Also note that again existing MPC techniques can be used in order to design the ingredients of problem \eqref{eq:stability_filter_tracking}, compare, e.g., with \cite{Koehler_reference_gen}, or Section \ref{sec:num_example}.

\section{Numerical Example}\label{sec:num_example}
   \begin{figure*}[t!] 
      \centering
        \includegraphics[scale=1]{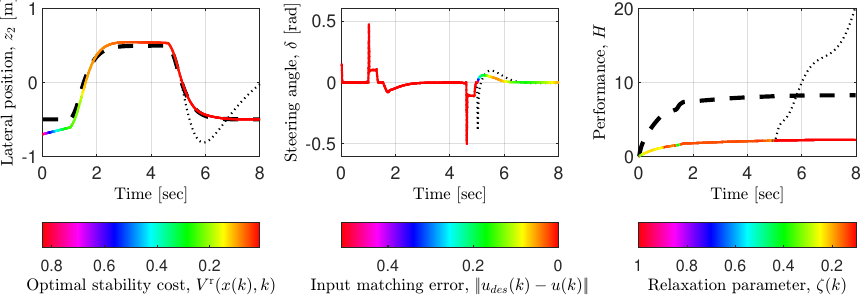}
        \caption{Numerical simulation of the stable safety filter to enhance a human driver with safety and stability guarantees during, e.g., obstacle avoidance. \textbf{Left:} Lateral position resulting from the desired input (black dotted) and the filtered closed-loop input (solid colorized), where the colors indicate the value of the current closed-loop stability cost. The reference is displayed by the black dashed line. \textbf{Center:} Desired steering angle (black dotted) and filtered closed-loop steering angle (solid colorized), where the colors indicate the value of the current input matching error. \textbf{Right:} Closed-loop performance resulting from the desired input (black dotted) and from the filtered closed-loop input (solid colorized), where the colors indicate the current value of the relaxation parameter. The stability bound $H_\mathrm{B}$ is displayed by the black dashed line.}
        \label{fig:example}
   \end{figure*}
We apply the proposed stability filter to an automotive driving scenario with two consecutive lane changes of a simulated human-operated car, illustrating, e.g., an overtaking or an obstacle avoidance maneuver.

The vehicle dynamics are described by a single-track model, see, e.g.\, \cite{CommonRoad}. The system states are $x = \begin{bsmallmatrix} z_1 & z_2 & \psi & v & \beta& \Dot{\psi} \end{bsmallmatrix}^\top \in \mathbb{R}^6$ and the inputs are $u=\begin{bsmallmatrix} \delta & a \end{bsmallmatrix}^\top \in \mathbb{R}^2$,
with the longitudinal position $z_1$, lateral position $z_2$, heading angle $\psi$, velocity $v$, side slip angle $\beta$, yaw rate $\Dot{\psi}$, steering angle $\delta$, and acceleration $a$, where all the states and inputs are considered relatively to the linearization point $(\Tilde{x},\Tilde{u}) = (\begin{bsmallmatrix} 0 & 0 & 0 & 10 & 0 & 0 \end{bsmallmatrix}^\top, \begin{bsmallmatrix}
    0 & 0
\end{bsmallmatrix}^\top)$. The linearized dynamics are discretized using a sampling time of $T_\mathrm{s}=0.02$ seconds. 
The system is subject to box constraints $x^\mathrm{l}~\leq~x~\leq x^\mathrm{u}$, $u^\mathrm{l}~\leq~u~\leq~u^\mathrm{u}$, with $x^\mathrm{l} = - x^\mathrm{u}$, $x^\mathrm{u} = [1,1,\frac{30 \pi}{180}, \frac{10}{3.6}, \frac{5 \pi}{180}, \frac{35 \pi}{180}]^\top$, $u^\mathrm{l}= [-\frac{35\pi}{180},-7]^\top$ and $u^\mathrm{u}=[\frac{35\pi}{180},2]^\top$.
We choose $Z^\mathrm{r} = \{(x^\mathrm{r},u^\mathrm{r}) \mid x^\mathrm{l}~\leq~1.05x^\mathrm{r}~\leq x^\mathrm{u}, u^\mathrm{l}~\leq~1.05u^\mathrm{r}~\leq~u^\mathrm{u} \}$. 

Despite constraint satisfaction, the goal is to achieve stability according to Definition \ref{defi:stability_tracking} with respect to a reference $r(k)$, provided, e.g., by a higher-level planner. To satisfy Assumptions~\ref{ass:tracking_cost_terms}  and \ref{ass:evol_stab_bound_tracking}, we use the design procedure outlined in \cite{Koehler_reference_gen} by choosing quadratic cost terms $\ell^\mathrm{r}(x,u,r):= \lVert x- x^\mathrm{r} \rVert_Q^2 + \lVert u - u^\mathrm{r} \rVert_R^2$, $V^\mathrm{r}_\mathrm{f}(x,x^\mathrm{r}):=\lVert x- x^\mathrm{r} \rVert_P^2$, and a terminal set of the from $\mathbb{X}_\mathrm{f}^\mathrm{r}(r) := \{x \in \mathbb{R}^n \mid V_\mathrm{f}^\mathrm{r}(x,x^\mathrm{r}) \leq \tau \}$ with $\tau > 0$. For $Q = I_6$ and $R = I_2$, we obtain $P \succ 0$ and a stabilizing state feedback matrix $K$, which allows us to construct a candidate sequence satisfying Assumption \ref{ass:evol_stab_bound_tracking} with $\rho = 1$  \cite{Koehler_reference_gen}:
    \begin{equation}\label{eq:candidate_input}
        u_{i | k+1}^\mathrm{c} = 
            \begin{cases} 
            u^{\mathrm{r},*}_{i+1 |k} & i \in \mathbb{I}_{[0,N-2]} \\
            \kappa_\mathrm{f}(x_{N|k}, r(k+N))  & i = N-1
            \end{cases},
    \end{equation}
    where $\kappa_\mathrm{f}(x_{N|k}, r(k+N)):=u^\mathrm{r}(k+N) + K(x_{N|k}-x^{\mathrm{r}}(k+N))$.
We select the planning horizon $N=30$ and start the maneuver with initial condition $x(0) = \begin{bsmallmatrix}
    0 & -0.5 & 0 & 0 &0 & 0
\end{bsmallmatrix}^\top$ with $J^\mathrm{r}_\mathrm{B}(0)=V^{\mathrm{r},\text{MPC}}(x(0),r_{.|0})$, which satisfies Assumption~\ref{ass:init_stab_bound_traj_track} for some $\alpha_2 \in \mathcal{K}_\infty$, compare with \cite{Koehler_reference_gen}.  The desired convergence is achieved using $ J^\mathrm{r}_{\mathrm{B}}(k) := V^\mathrm{r}(x(k-1),k-1) - \zeta(k) \ell^\mathrm{r}(x(k-1),u^{\mathrm{r},*}_{0|k-1},r_{0|k-1})$ for all $k \in \mathbb{I}_{>0}$, where we add $\zeta(k) \in [\zeta_{\mathrm{min}},1]$ with $\zeta_{\mathrm{min}} = 0.1$ as an optimization variable to problem \eqref{eq:stability_filter_tracking}. With the latter, we aim at improving the input matching behavior while at the same time ensuring that $H_\mathrm{B}(k) := \frac{V^\mathrm{r}(x(0),0)- V^\mathrm{r}(x(k),k)}{\zeta_{\mathrm{min}}}$ provides an upper bound on the closed-loop performance measured by $H(k) := \sum_{i=0}^{k} \ell^\mathrm{r}(x(i),u(i),r(i))$ for all $k \in \mathbb{I}_{\geq 0}$. 
   
The results are presented in Fig. \ref{fig:example}. For time $t = k\cdot T_{\mathrm{s}}\in[0,\;5]\mathrm{s}$, the filter simultaneously matches the inputs and tracks the reference according to the stability and performance bound. The relaxation parameter $\zeta(k)$ is selected to $\zeta(k) = 0.18$ in average for $t\in[0,\;5]\mathrm{s}$, allowing the desired inputs to be directly applied by the filter. Even very harsh steering angles, e.g., desired at $t= 4.6\mathrm{s}$, are passed by the filter. For $t > 5\mathrm{s}$, the desired inputs would lead to an unstable behavior which causes the stability filter to intervene in order to ensure closed-loop stability and the required performance. However, the filter still tries to be as minimally invasive as possible while ensuring all requirements are met, which is indicated by the fact that $\zeta(k)=0.1$ for $t > 5\mathrm{s}$.

\section{Conclusion}
We have extended predictive safety filters to ensure different levels of stability and linked the design procedures to well-known MPC techniques. Thus, we provide a modular framework that allows to enhance any potentially unsafe and unstable control strategy to be augmented with safety and stability guarantees. This extends the applicability of predictive safety filters to a wide range of practical problems, ranging from system identification tasks to setpoint stabilization problems and dynamic reference stabilization settings. We have demonstrated the applicability of the framework considering an advanced driver assistance setting in simulation.

\bibliographystyle{plain}
\bibliography{References}

\end{document}